\documentclass[11pt]{article}

\usepackage{fullpage}

\usepackage[utf8]{inputenc}
\usepackage[english]{babel}
\usepackage[compact]{titlesec}
\usepackage{setspace}

\usepackage{latexsym, amssymb, bbm}
\usepackage{amsmath}
\usepackage{amsthm}
\usepackage{verbatim}
\usepackage{enumitem}

\usepackage[dvipsnames]{xcolor}
\usepackage[breaklinks]{hyperref}
\hypersetup{colorlinks=true,
            citebordercolor={.6 .6 .6},linkbordercolor={.6 .6 .6},
citecolor=blue,urlcolor=black,linkcolor=blue,pagecolor=black}
\usepackage[nameinlink]{cleveref}

\usepackage{mathtools}
\usepackage{thmtools,thm-restate}

\usepackage{graphicx}
\usepackage{nicefrac,xspace}
\usepackage{authblk}

\usepackage{algorithm}
\usepackage[noend]{algpseudocode}

\newtheorem{theorem}{Theorem}[section]
\newtheorem{lemma}[theorem]{Lemma}
\newtheorem{proposition}[theorem]{Proposition}
\newtheorem{cor}[theorem]{Corollary}
\newtheorem{fact}[theorem]{Fact}
\newtheorem{definition}[theorem]{Definition}

\theoremstyle{remark}

\newcommand{\cut}[1]{}

\newcommand{\eps}{\varepsilon}

\def\N{\mathbb{N}}

\def\R{\mathbb{R}}

\def\d{\delta}
\def\e{\varepsilon}

\def\cB{\mathcal{B}}

\def\cK{\mathcal{K}}

\def\la{\langle}
\def\ra{\rangle}
\newcommand{\ip}[2]{{\langle #1,#2 \rangle}}

\def\st{\text{st}}

\def\sse{\subseteq}

\renewcommand{\bar}[1]{\overline{#1}}

\newcommand{\argmax}{\text{argmax}}

\def\blue{\color{blue}}

\renewcommand{\emptyset}{\varnothing}

\newcommand{\initOneLiners}{%
    \setlength{\itemsep}{0pt}
    \setlength{\parsep }{0pt}
    \setlength{\topsep }{0pt}
}

\newcommand{\CBC}{\textsf{CBC}\xspace}
\newcommand{\ang}{\rho}
\renewcommand{\d}{\textrm{d}}
\newcommand{\grad}{\nabla}

\title{Lipschitz Selectors may not Yield Competitive Algorithms for Convex Body Chasing}
\author{C.J. Argue, Anupam Gupta, and Marco Molinaro}
\date{\today}

\begin{document}
\maketitle

\begin{abstract}
  The current best algorithms for convex body chasing problem in
  online algorithms use the notion of the Steiner point of a convex
  set. In particular, the algorithm which always moves to the Steiner
  point of the request set is $O(d)$ competitive for nested convex
  body chasing, and this is optimal among memoryless algorithms
  \cite{bubeck2020chasing}. A memoryless algorithm coincides with the
  notion of a selector in functional analysis. The Steiner point is
  noted for being Lipschitz with respect to the Hausdorff metric, and
  for achieving the minimal Lipschitz constant possible. It is natural
  to ask whether every selector with this Lipschitz property yields a
  competitive algorithm for nested convex body chasing. We answer this
  question in the negative by exhibiting a selector which yields a
  non-competitive algorithm for nested convex body chasing but is
  Lipschitz with respect to Hausdorff distance. Furthermore, we show
  that being Lipschitz with respect to an $L_p$-type analog to the Hausdorff
  distance is sufficient to guarantee competitiveness if and only if $p=1$.
\end{abstract}

\section{Introduction}

In the convex body chasing (\CBC) problem, the player receives a sequence of nonempty closed convex sets $K_1,K_2, \dots, K_T\subseteq \R^d$ and must respond to each set $K_t$ with a point $x_t\in K_t$. Moreover, the points $x_t$ are selected in an online fashion; that is, $x_t$ must be fixed before $K_{t+1}$ is revealed. The objective is to minimize the total distance traveled, i.e. $\sum_{t=1}^T \|x_t - x_{t-1}\|$, where $x_0 = 0$. The performance of an algorithm is measured by the \emph{competitive ratio}, defined as
\begin{equation}
\sup_{T\in \N} \sup_{K_1,\dots, K_T} \frac{ALG(K_1,\dots, K_T)}{OPT(K_1,\dots, K_T)} \label{eq:competitive}
\end{equation}
where $ALG(K_1,\dots, K_T)$ denotes the cost of the algorithm's
solution and $OPT(K_1,\dots, K_T)$ denotes the cost of the
optimal solution in hindsight. A \emph{competitive} algorithm is one
with finite competitive ratio (i.e. the ratio above is uniformly bounded for every sequence of nonempty convex sets).

In this paper, we consider \emph{memoryless} algorithms for \CBC,
namely algorithms such that $x_t$ depends only on $K_t$, and not on
previous request sets or points chosen. That is, before seeing any
$K_i$, the algorithm fixes a function $s:\{$nonempty convex subsets of
$\R^d\}\to \R^d$ such that $s(K)\in K$ for all $K$. Such a function
$s$ is known as a \emph{selector}. Henceforth, the algorithm defines
$x_t = s(K_t)$.

Selectors cannot give competitive algorithms for \CBC in general.\footnote{Fix a selector $s$ and let $K$ and $K'$ be two sets such that $s(K)\neq s(K')$ and $K\cap K'\neq \emptyset$. (One can find such a pair of sets, e.g., among the edges of a triangle.) Now consider the request sequence $K, K', K, K', \dots$. The offline optimum is finite, as one can choose $x\in K\cap K'$ and let $x_t=x$ for all $t$. However, the online algorithm pays $\|s(K)-s(K')\|$ at each step, so its cost is unbounded.} However, for \emph{nested} instances where $K_1\supseteq K_2\supseteq \dots \supseteq K_T$, selectors can give competitive algorithms. Bubeck et al. show that \emph{Steiner point selector} $\st(K)$, which is defined as the average of an extreme point of $K$ in a uniformly random direction, achieves the optimal competitive ratio among selectors \cite{bubeck2020chasing}.

The Steiner point is well-known in the functional analytic study of \emph{Lipschitz selectors}. Namely, the Steiner point is notable for being $O(\sqrt{d})$-Lipschitz with respect to the \emph{Hausdorff distance} between two sets, defined as the maximum distance from a point in one set to the other set \cite{daugavet1968some,vitale1985steiner}. In fact, the Steiner point achieves the minimal Lipschitz constant among all selectors \cite{przeslawski1989continuity}. Steiner-type selectors exist even when $\R^d$ is replaced by a general Banach space~\cite{SHVARTSMAN20041}. Lipschitz selections (of set-valued functions) are well-studied and have variety of applications to differential inclusions, metric projection,  calculus of variations, etc. (see \cite{shvartsman} and references therein).

Given the Steiner point's properties as both a competitive and
Lipschitz selector, it is natural to ask about the connection between
these two properties. In particular, is every Lipschitz selector
also competitive for nested \CBC? We consider this question for a
broader class of metrics $\{D_p\}_{p\ge 1}$ for sets, which are $L_p$-analogs
of the Hausdorff metric. These were first defined
in~\cite{mcclure1975polygonal}; $D_\infty$ is the Hausdorff metric. We define them in Section~\ref{sec:def}. Our main result answers the question in the negative. 

\begin{theorem}\label{thm:main-intro}
For any $p \in (1,\infty]$, there is a selector that is Lipschitz with respect to $D_p$ but is not competitive for nested \CBC.
\end{theorem}

Furthermore, this result is tight in the sense that it cannot be extended to $p=1$. Indeed, we note that every selector that is $L$-Lipschitz with respect to the $D_1$ metric is actually $L$-competitive. This fact is implicit in Bubeck et al.'s proof that the Steiner point is $d$-competitive~\cite{bubeck2020chasing}. For completeness, we include a proof in \Cref{app:SteinerD1}.

The proof of Theorem~\ref{thm:main-intro} has three main steps.
\begin{enumerate}
\item Construct a ``hard'' nested sequence of sets $K_1\supseteq K_2\supseteq \dots$ such that $\sum_t D_p(K_t, K_{t+1}) = \infty$. Notice that the existence of such a sequence is necessary in order to have a selector that is both non-competitive and Lipschitz with respect to $D_p$. 

\vspace{6pt}
\item Define a partial selector $\bar{s}: \{K_t\}_{t\in \N}\to \R^d$ for these set that is Lipschitz with respect to $D_p$ but not competitive on this sequence.

\vspace{6pt}
\item Extend the partial selector $\bar s$ to a full selector $s: \{$nonempty convex subsets of $\R^d\}\to \R^d$ that is also Lipschitz with respect to $D_p$ (with a larger Lipschitz constant than that of $\bar s$). The selector $s$ is also non-competitive, since it inherits this property from $\bar{s}$.

\smallskip
We remark that there are well-known techniques for extending a Lipschitz partial function to a Lipschitz function on the entire domain; the difficulty here lies in ensuring that the extension is a selector, that is, that $s(K)\in K$ for all $K$.
\end{enumerate}
We construct the sequence of sets in \Cref{sec:packing}, the partial
selector in \Cref{sec:partial-selector}, and the full selector in \Cref{sec:full-selector}.

\subsection{Related work}

Friedman and Linial introduced \CBC in 1993 in order to study the geometry of the related Metrical Task Systems problem \cite{friedman1993convex}. They proved that no algorithm for \CBC in dimension $d$ has competitive ratio lower than $\sqrt{d}$, and gave a competitive algorithm when $d=2$. Initially, most work on \CBC focused on special cases; see, e.g.~\cite{fujiwara2008online,sitters2014generalized,antoniadis2016chasing}. A recent series of papers gave competitive algorithms for the nested case, where $K_1\supseteq K_2\supseteq \dots \supseteq K_T$ \cite{bansal2019nested,argue2019nearly,bubeck2020chasing}. These solutions led to the first competitive algorithm for the general case of \CBC \cite{bubeck2019competitively,argue2020chasing,sellke2020chasing}. The current best-known bounds are $O(\sqrt{d\log d})$ for the nested case \cite{bubeck2020chasing}, $O(\min(d, \sqrt{d\log T}))$ for the general case \cite{argue2020chasing,sellke2020chasing} and $O(\min(k, \sqrt{k\log T}))$ for $k$-dimensional subspaces \cite{argue2020dimension}. We remark that the algorithms of~\cite{argue2020chasing,sellke2020chasing} for the general case also rely on a Steiner-type point, although now applied to a modified version of the sets $K_t$.

\subsection{Definitions and notation}\label{sec:def}

All norms for points in $\R^d$ in this paper are Euclidean. We let $B^d(x_0,r) := \{x\in \R^d: \|x-x_0\| \le r\}$ denote the ball of radius $r$ centered at $x_0\in\R^d$ and we let $B^d:=B(0,1)$ denote the unit ball in $\R^d$. For convenience, we omit the superscripts and write $B(x_0,r)$ and $B$. We let $S^{d-1} = \{x\in \R^d: \|x\|=1\}$ denote the unit sphere in $\R^d$. 

We denote by $\cK^d$ the set of nonempty compact convex sets in $\R^d$, and write $\cK$ when $d$ is clear from context. The \emph{support function} of a set $K\in \cK^d$ is the function $h_K:S^{d-1}\to \R$ given by
\[h_K(y) = \max_{x\in K} \ip{x}{y}.\]
A convex set $K$ is \emph{centrally symmetric} if $K = -K:= \{-x: x\in K\}$.

Let $\sigma_d$ denote the Lebesgue measure on $S^{d-1}$ normalized so that $\sigma_d(S^{d-1}) = 1$. We henceforth drop the subscript and write $\sigma$. The \emph{Hausdorff distance} between two sets in $\cK^d$ is given by the $L_\infty$-distance between their support functions:
\[D_\infty(K,K') := \max_{y \in S^{d-1}} |h_K(y) - h_{K'}(y)|,\] 
More generally, for any $p\in [1,\infty)$, the metric $D_p$ between two sets in $\cK^d$ is given by the $L_p$-distance between their support functions:
\[D_p(K,K') = \left(\int_{y \in S^{d-1}} |h_K(y) - h_{K'}(y)|^p \ d\sigma(y) \right)^{1/p}.\]

For a collection $\bar{\cK}\subseteq \cK^d$ of compact convex sets, a \emph{partial selector} $\bar{s}$ is a function $\bar{s}:\bar{\cK} \to \R^d$ such that $s(K)\in K$ for all $K\in \bar{\cK}$. When $\bar{\cK} = \cK^d$, we say $s = \bar{s}$ is a \emph{selector}. The (partial) selector $\bar{s}:\bar{\cK}\to \R^d$ is $L$-\emph{Lipschitz with respect to $D_p$} if 
\begin{equation}
\|\bar{s}(K_1) - \bar{s}(K_2)\| \le L \cdot D_p(K_1,K_2) \qquad \forall K_1,K_2\in \bar{\cK}.
\end{equation}
We now have the definition of one of the central concepts in this paper.

\begin{definition}[Competitive selector]
For $\gamma > 0$, a (partial) selector $\bar{s} : \bar{\cK} \rightarrow \R^d$ is \emph{$\gamma$-competitive} if every infinite nested sequence $B \supseteq K_1
\supseteq K_2\supseteq \dots$ of sets $K_t\in\bar{\cK}$ satisfies
$\sum_{t=1}^\infty \|\bar{s}(K_{t+1})-\bar{s}(K_t)\| \le
\gamma.$
\end{definition}

Notice that this is not exactly the same as the definition of a competitive algorithm for nested \CBC, since it imposes the condition that all sets $K_t$ are inside the unit ball $B$ and it replaces the competitive ratio \eqref{eq:competitive} by just the total movement $\sum_{t=1}^\infty \|\bar{s}(K_{t+1})-\bar{s}(K_t)\|$. However, a simple reduction shows that a $\gamma$-competitive selector yields a $O(\gamma)$-competitive algorithm for nested \CBC, see~\cite[Lemma 1]{bansal2019nested}.

A selector that is not $\gamma$-competitive for any $\gamma > 0$ is said to be \emph{non-competitive}. Note that any extension of a non-competitive selector is also non-competitive.

The \emph{Steiner point selector} $\st:\cK^d\to \R^d$ is defined by
\begin{equation}\label{def:st-1}
\st(K) := \int_{y\in S^{d-1}} \grad h_K(y)\ d\sigma(y),
\end{equation}
or equivalently,
\begin{equation}\label{def:st-2}
\st(K) := d\int_{y\in S^{d-1}} y \cdot h_K(y)\ d\sigma(y),
\end{equation}
see page 315 of \cite{schneider} for the equivalence.
Note that since $\grad h_K(y) = \argmax_{x\in K} \ip{x}{y}$, definition~(\ref{def:st-1}) implies that $\st(K)\in K$, i.e. the Steiner point is in fact a selector. It follows easily from definition~(\ref{def:st-2}) that $\st(\cdot)$ is $d$-Lipschitz with respect to $D_1$. Notice that for any $p\ge 1$ we have $D_1(K,K')\le D_p(K,K')$. Consequently, $\st(\cdot)$ is $d$-Lipschitz with respect to $D_p$, for all $p\ge 1$.

For two points $u,v\in S^{d-1}$, we define their distance $\ang(u,v)$ as the angle between them, i.e. $\ang(u,v) = \arccos \ip{u}{v}$. Notice this is exactly the geodesic distance on the sphere, so it satisfies triangle inequality.
For $u\in S^{d-1}$ and $\theta\in (0,\pi]$, let $C(u,\theta)$ be the
cap on the sphere $S^{d-1}$ around $u$ with angle $\theta$, namely 
$$C(u,\theta) := \{x \in S^{d-1} : \ang(x,u) < \theta\}.$$ 
Equivalently, $C(u,\theta)$ is the open ball of radius $\theta$ around $u$ with respect to the distance $\ang$.  
Let $\bar{C}(u,\theta)$ denote the convex hull of cap $C(u,\theta)$. Let $C_\pm(u,\theta) = C(u,\theta) \cup C(-u,\theta)$ and let $\bar{C}_\pm(u,\theta) = \bar{C}(u,\theta) \cup \bar{C}(-u,\theta)$; we refer to these as \emph{double caps}. Note that $C(u,\theta)$ is open as a subset of $S^{d-1}$, $\bar{C}(u,\theta)$ is open as a subset of $B$, and hence for any compact set $K\subseteq B$ it holds that $K- \bar{C}(u,\theta)$ is compact.


\section{A nested sequence with infinite $D_p$ movement}
\label{sec:packing}

We now construct, for each $p>1$, a nested sequence of sets such that the cumulative distance between consecutive sets (as measured by $D_p$) is infinite. 

\begin{proposition}\label{prop:hard-instance}
Fix any $p > 1$, dimension $d\ge \frac{4p-1}{p-1}$, and $\eps > 0$. Then there is a sequence $B\supseteq K_1\supseteq K_2\supseteq \dots $ of centrally symmetric sets in $\cK^d$ such that
\begin{itemize}
\item[(I)] $D_\infty(B,K_i) \le \eps$ for all $i$.
\item[(II)]$\sum_{i=1}^\infty D_p(K_i,K_{i+1}) = \infty$.
\item[(III)] The distances $D_p(K_i, K_{i+1})$ are decreasing for $i = 1,2,\ldots$.
\end{itemize}
\end{proposition}

Notice that in addition to having infinite cumulative distance (Item (II)), we also have the property that the sets $K_i$ are very close to being a ball (Item (I)). This will be important in the construction of our partial selector $\bar{s} : \{K_i\} \rightarrow \R^d$ and its extension to a selector $s$ for all sets in $\cK$. As mentioned in the introduction, guaranteeing $s(K) \in K$ is the hard part of obtaining an extension. Item (I) allows us to place $\bar{s}(K_i)$ ``very deep'' into $K_i$, so that for a set $K \approx K_i$ it suffices for the extension to have $s(K) \approx s(K_i) = \bar{s}(K_i)$ to ensure that $s(K)$ is inside $K$. (Item (III) in the statement is mostly a convenient technical property.)

This infinite sequence of sets $K_i$ is obtained by starting from the unit ball $B$ and removing an infinite set of \textbf{disjoint} double caps $\bar{C}_{\pm}(u_i, \theta_i)$, one at a time. Disjointness will be useful for two reasons. First, it means that when removing the double cap $\bar{C}_{\pm}(u_i, \theta_i)$ from $K_i$ (to obtain the next set $K_{i+1}$) the part lost is exactly the same as if removing this double cap from the original ball, which greatly simplifies computing the distance $D_p(K_i, K_{i+1})$. Second, it more easily guarantees that we never cut in too deeply (ensuring Item (I) of the proposition). 

	To carry out this idea, we need to show that there is indeed an infinite set of these double caps that are (a) ``small enough'' to be disjoint but that (b) ``cut enough'' to make sure the cumulative distance $\sum_{i=1}^\infty D_p(K_i,K_{i+1})$ is infinite (Item (II) of the proposition).\footnote{The underlying reason why this is possible for $p > 1$ (but not $p=1$), is the following: If the body $K'$ is obtained by cutting off a small part of $K$, then $\frac{D_p(K',K)}{D_1(K',K)}$ goes to infinity as the size of the piece cut off goes to 0. Thus, if one chooses an infinite sequence of cuts where the sum of the $D_1$-distance between the subsequent bodies converges but barely, then the sum of the $D_p$-distances will diverge. This phenomenon is flexible enough to also accommodate Items I and III of the proposition.} Our key technical ingredients are (respectively) (a) estimates for the normalized area of a cap and (b) the  $D_p$-distance between the unit ball and the ball minus a cap. 

 To start, recall that $\sigma$ is the normalized uniform measure over $S^{d-1}$. We abuse notation and write $\sigma(C(\theta))$ for the area of any cap with angle $\theta$ (the area depends only on the angle, and not on the center).

Our first lemma is an upper bound on $\sigma(C(\theta))$ in terms of $\theta$. This follows simply from the following (slightly restated) estimate that can be found, for example, in~\cite{cai2013distributions}.

\begin{fact}{\cite[Proposition 5]{cai2013distributions}}\label{fact:prob-bd-1}
Fix $u\in S^{d-1}$ and $\lambda \in(0,\frac{\pi}{2})$. Then the probability that a uniformly 
random point from $S^{d-1}$ has an angle with $u$ that is outside of the interval $(\frac{\pi}{2} - \lambda, \frac{\pi}{2} + \lambda)$
is at most $\gamma\cdot \sqrt{d}\cdot \cos(\lambda)^{d-2}$, where $\gamma$ is a constant that does not depend on $d$ or $\lambda$.
\end{fact}

\begin{lemma}\label{lem:area-cap}
For $\theta \in (0, \frac{\pi}{2})$, 
$$\sigma(C(\theta)) \le \gamma \cdot \sqrt{d} \cdot \theta^{d-2},$$ 
where $\gamma$ is a constant that does not depend on $d$ or $\theta$.
\end{lemma}
	  
\begin{proof}  
 Fix any center $u$ for the cap, and letting $\lambda = \frac{\pi}{2} - \theta$ we see that the area $\sigma(C(\theta))$ is at most the probability bounded in \Cref{fact:prob-bd-1}. The conclusion then follows from the estimate $\cos(\lambda) = \cos(\frac{\pi}{2}-\theta) =
\sin(\theta) \le \theta$, which is valid for all $\theta \ge 0$.
 \end{proof}
	
Using this estimate, we can show how to pack double-caps into a sphere using a standard volumetric bound. 
	
\begin{lemma}\label{lem:disj-caps}
Let $\frac{\pi}{2}\ge \theta_1 \ge \theta_2 \ge \ldots$ be such that $\sum_{i=1}^\infty \sigma(C(2\theta_i)) < \frac{1}{2}$. Then there is a set of points $u_1,u_2,\ldots$ in $S^{d-1}$ such that the double-caps $\{C_{\pm}(u_i, \theta_i)\}_{i\in \N}$ are pairwise disjoint.
\end{lemma}
	
\begin{proof}
The claim is equivalent to saying that for all $i <j$, it holds that $u_j \notin C_\pm(u_i, \theta_i+\theta_j)$. Since the $\theta_i$ are decreasing, $\theta_i +\theta_j < 2\theta_i$. Therefore it suffices to choose $u_j\notin C_{\pm} (u_i, 2\theta_i)$ for any $i<j$. We construct such $u_i$'s iteratively: after picking $u_1,\ldots,u_i$, we have 
$$
\sum_{i=1}^k \sigma(C_{\pm} (u_i, 2\theta_i)) 
~=~ 2 \sum_{i=1}^k \sigma(C (2\theta_i))
~<~ 1 ~=~ \sigma(S^{d-1}).
$$
Thus there is a point in $S^{d-1}$ that is not in any previous cap. Choose this point to be $u_{i+1}$.   \end{proof}
	
Using \Cref{lem:area-cap}, for the packing property to hold we just need 
\begin{align}
	\sum_i \theta_i^{d-2} \le g(d,p), \label{eq:packingNeed}
\end{align}
for some function $g$ of $d$ and $p$. 

We now turn to getting a lower bound on $D_p(B, B - \bar{C}(u_i,\theta_i))$ in terms of $\theta$. 
It will be convenient to have another technical lemma also from~\cite{cai2013distributions}.

\begin{fact}{\cite[Lemma 12]{cai2013distributions}}\label{fact:prob-bd-2}
Let $v\in S^{d-1}$ be fixed and $U$ be sampled uniformly from $S^{d-1}$. Then the random angle $\rho(U,v)$ between $U$ and $v$ is distributed according to the density function $t \mapsto f_1(d) (\sin t)^{d-2}$, where $f_1(d)$ is a constant depending only on $d$.
\end{fact}

Note that the following lemma and its proof consider caps $\bar{C}(u_i,\theta_i)$ and not double-caps $\bar{C}_\pm(u_i,\theta_i)$.
	
\begin{lemma}\label{lem:dp-large}
Let $\theta\in(0,\frac{\pi}{2}]$, and $v\in S^{d-1}$ be arbitrary. Then
$$D_p(B, B - \bar{C}(v,\theta)) \ge f(d,p) \cdot \theta^{2 + (d-1)/p}$$ for some function $f$ of $d$ and $p$. 
\end{lemma}	
	
\begin{proof}
For notational ease, let $B^-:= B - \bar{C}(v,\theta).$
Firstly, $h_{B}$ and $h_{B^-}$ only differ on directions $u \in C(v, \theta)$. For such $u$, note that $C(u, \theta - \ang(v,u))$ is contained in $C(v, \theta)$ because $\ang(\cdot,\cdot)$ satisfies the triangle inequality. 
Thus, $B^{-}$ has excluded all points from $C(u, \theta - \ang(u,v))$, that is, every vector $y \in B^{-}$ has angle at least $\theta - \ang(u,v)$ with $u$, and hence $\ip{u}{y} \le \cos(\theta - \ang(u,v)$. Thus,
$$h_{B^-}(u) 
	\le \cos\bigg(\theta - \ang(v, u)\bigg) 
	\le 1 - \frac{1}{4}\bigg(\theta - \ang(v, u)\bigg)^2,$$
where the second inequality follows from the estimate $\cos(x) \le 1 - \frac{x^2}{4}$, which is valid for all $x \in [0,\frac{\pi}{2}]$.
Since $h_{B}(u) = 1$ for all $u \in S^{d-1}$, using the above bound and then \Cref{fact:prob-bd-2} we obtain
\begin{align*}
	D_p(B, B^-)^p &= \int_{C(v,\theta) \Big(h_B(u) - h_{B^-}(u) \Big)^p ~d\sigma(u)}\\
	&\ge 	\int_{C(v,\theta)} \bigg[\frac{1}{4}\Big(\theta - \ang(v,u)\Big)^2\bigg]^p ~d\sigma(u)\\
	&= f_1(d) \cdot \int_0^{\theta} \bigg[\frac{1}{4}\Big(\theta - t\Big)^2\bigg]^p \cdot (\sin t)^{d-2} ~\d t \tag{\Cref{fact:prob-bd-2}}\\
	&\ge f_2(d,p) \cdot \int_0^{\blue{\theta/2}} \Big(\theta - \tfrac{\theta}{2}\Big)^{2p} \cdot (\sin t)^{d-2} ~\d t \tag{$f_2 := \frac{1}{4^p}\cdot f_1$}\\
	&\ge f_3(d,p) \cdot \int_0^{\theta/2} \Big(\tfrac{\theta}{2}\Big)^{2p} \cdot \blue{t}^{d-2} ~\d t \tag{$f_3 := \frac{1}{2^{d-2}}\cdot f_2$}\\
	& = f_4(d,p) \cdot \theta^{2p+d-1}, \tag{$f_4 := \frac{1}{2^{2p+d-1}(d-2)}\cdot f_3$} 
\end{align*}
where in the last inequality we used that $(\sin t) \ge \frac{t}{2}$ for $t\le \frac{\pi}{4}$. Taking $p$-th roots and renaming $f_4$ proves the lemma.  
 \end{proof}
	
\Cref{lem:disj-caps} and~\ref{lem:dp-large} now give a sequence of disjoint double-caps that induce infinite movement.

\begin{lemma} \label{lem:packing}
Fix any $p > 1$, dimension $d\ge \frac{4p-1}{p-1}$, and $\eps > 0$. Then there is a sequence of angles $\arccos(1-\eps) \ge \theta_1\ge \theta_2 \ge \dots$ and a sequence  $u_1,u_2,\dots$ of points in $S^{d-1}$ satisfying:
\begin{enumerate}\setlength{\itemsep}{0pt}
	\item[(i)] The double-caps $\{C_{\pm}(u_i, \theta_i)\}_{i\in \N}$ are pairwise disjoint.
	\item[(ii)] $\sum_{i=1}^\infty D_p(B, B - \bar{C}(u_i,\theta_i)) = \infty$
\end{enumerate}
\end{lemma}
	
\begin{proof}
Define 
\[\theta_i := \alpha \bigg(\frac{1}{i}\bigg)^{1/(2 + (d-1)/p)},\]
where $\alpha>0$ is a small constant (depending on $d$ and $p$) to be specified later.
Since $d > \frac{4p-1}{p-1}$, it follows that $\frac{d-2}{2+(d-1)/p} > 1$, whence using \Cref{lem:area-cap}
$$\sum_i \sigma(C(2\theta_i)) \le 2^{d-2}\gamma\sqrt{d}\cdot \sum_i \theta_i^{d-2} < \infty.$$
By choosing $\alpha$ small enough, this sum can be made less than $\frac{1}{2}$. Then by \Cref{lem:disj-caps} we can find $u_i$'s such that the double caps $\{C_{\pm}(u_i,\theta_i)\}_{i\in \N}$ satisfy condition $(i)$. By \Cref{lem:dp-large}, we have
\[\sum_{i=1}^\infty D_p(B, B - \bar{C}(u_i,\theta_i))
	\ge \sum_{i=1}^\infty f(d,p) \cdot \theta_i^{2+(d-1)/p}
	= f(d,p)\cdot \alpha^{2+(d-1)/p} \sum_{i=1}^\infty \frac{1}{i}
	= \infty.\]
Hence $(ii)$ is satisfied.
\end{proof}

\begin{proof}[of \Cref{prop:hard-instance}]
Finally we  prove \Cref{prop:hard-instance}. With sequences $\theta_1, \theta_2, \dots$ and $u_1,u_2,\dots$ as given by \Cref{lem:packing}, define the nested sequence $K_1 \supseteq K_2 \supseteq \dots$ of convex bodies by starting from the unit ball $B$ and sequentially removing the double-caps $\bar{C}_{\pm}(u_i, \theta_i)$:
\begin{equation}
\begin{tabular}{r c l}
$K_1$ & $=$ & $B$, \\
$K_{i+1}$ & $=$ & $K_i - \bar{C}_{\pm}(u_i,\theta_i), \qquad \forall i\ge 1$
\end{tabular} \label{eq:def-of-Kt}
\end{equation}
It is clear that each $K_i$ is centrally symmetric. For all $i$, 
\[D_\infty(B, K_i) \,\le\, \max_{j < i}\ (1-\cos(\theta_j)) \,\le\, \eps,\]
and so we obtain \Cref{prop:hard-instance}(I).

Moreover, since the double-caps are disjoint by \Cref{lem:packing}($i$), for any direction $y$ it holds that 
\begin{equation*}
h_{K_i}(y) - h_{K_{i+1}}(y) 
	~=~ h_{B}(y) - h_{B - \bar{C}_\pm (u_i,\theta_i)}(y),
\end{equation*}
because if $y \in \bar{C}_\pm (u_i,\theta_i)$ then $h_{K_i}(y) = h_{B}(y)$ and $h_{K_{i+1}}(y) =h_{B - \bar{C}_\pm (u_i,\theta_i)}(y)$, and otherwise both sides are zero.
Therefore, 
\begin{equation}
D_p(K_i, K_{i+1}) = D_p(B, B - \bar{C}_\pm (u_i,\theta_i)) = 2 D_p(B, B - \bar{C} (u_i,\theta_i)). \label{eq:distKs}
\end{equation}
It now follows from \Cref{lem:packing}$(ii)$ that $\sum_i D_p(K_i,
K_{i+1}) = \infty$, thus giving \Cref{prop:hard-instance}(II). 

The last item \Cref{prop:hard-instance}(III), namely that the distances $D_p(K_i,K_{i+1})$ are decreasing over $i$, also follows from \eqref{eq:distKs} and the fact the $\theta_i$'s are also decreasing over $i$. This completes the proof of
\Cref{prop:hard-instance}.
\end{proof}


\section{Defining the Selector}
\label{sec:extension}

Throughout this section, we fix $p > 1$, dimension $d\ge \frac{4p-1}{p-1}$, and $\e\in(0,1)$ to be specified later ($\e$ will be a function of $d$ and $p$). We will need the following lemma:

\begin{lemma}\label{lem:contains-ball}
Let $K\sse B$ be a nonempty convex compact set, and let $\rho:= D_\infty(K,B)$. Then $K$ contains the ball $B(0,1-\rho)$.
\end{lemma}

\begin{proof}
Before stating the proof we recall that the Hausdorff distance $D_{\infty}$ between $B$ and $K$ is also given by the classic formula (simplified since $K \subseteq B$) 
\begin{align}
	D_{\infty}(K,B) = \max_{x \in B} \min_{y \in K} \|x - y\|, \label{eq:hausEquiv}
\end{align} 
see Lemma 1.8.14 of~\cite{schneider}.

Now we prove the lemma in the contrapositive: if $K$ does not contain $B(0,1-\rho)$ then $D_\infty(K,B) > \rho$. Suppose there exists a point $z$ in $B(0,1-\rho)$ but not in $K$. Then there is a halfspace $\{x : \la u,x\ra \le \alpha\}$ that contains $K$ but not $z$. By rescaling, we can assume without loss of generality that $\|u\| = 1$.

Using the fact $z$ does not belong to this halsfpace, Cauchy-Schwarz inequality, and then that $z \in B(0,1-\rho)$, we have
\[\alpha < \la u,z \ra \le \|z\| \le 1-\rho.\]
Thus, for every $y\in K$ we have $\la u,y\ra \le \alpha < 1- \rho$, and so Cauchy-Schwarz inequality gives 
\[ \|u - y\| \ge \la u - y, u \ra > \rho.\]
Since $u\in B$, this implies that $\max_{x \in B} \min_{y \in K} \|x-y\| > \rho$, and from \eqref{eq:hausEquiv} we get $D_{\infty}(K,B) > \rho$ as desired. 
\end{proof}

\subsection{Partial selector}
\label{sec:partial-selector}

Let the sequence $K_1,K_2,\dots$ be as defined in~\eqref{eq:def-of-Kt},
which satisfies the properties in \Cref{prop:hard-instance}.
\begin{proposition}\label{prop:partial-selector}
For any $\eta>0$, there is partial selector $\bar{s}:\{K_i\}_{i\in \N}\to \R^d$ that is $2\eta$-Lipschitz with respect to $D_p$ and non-competitive.
\end{proposition}
\begin{proof}
Since an $\eta$-Lipschitz selector is also $\eta'$-Lipschitz for any $\eta'> \eta$, we may assume w.l.o.g.\ that $\eta < \frac{1-\e}{2\e}$. For $i \ge 1$, let $\alpha_i := D_p(K_i, K_{i+1})$. Now define a partial selector $\bar{s}: \{K_i\}_{i\in\N}\to \R^d$ by
\begin{equation}\label{eq:partial-selector}
\bar{s} (K_i) =
\left\lbrace
\begin{tabular}{cl}
$\eta \alpha_i \cdot e_1$ & if $i$ is odd\\
$-\eta \alpha_i \cdot e_1$ & if $i$ is even\\
\end{tabular}
\right. 
\end{equation}
Note that $\bar{s}$ is indeed a partial selector: by \Cref{prop:hard-instance}(I) we have $D_\infty(K_i,B)\le \eps$ for all $i$. This implies that each $K_i$ contains the ball $B(0,1-\e)$ (by Lemma~\ref{lem:contains-ball}) and that $\alpha_i \le 2\e$ (by triangle inequality). By the assumption that $\eta < \frac{1-\e}{2\e}$, it follows that $\|\bar{s}(K_i)\| < 1-\e$, whence $\bar{s}(K_i) \in K_i$.

\paragraph{Lipschitz.}
Fix $i < j$. Since the sets $K_i$ are nested we have $D_p(K_i, K_j) \ge D_p(K_i, K_{i+1}) = \alpha_i$. Also, as guaranteed by \Cref{prop:hard-instance}(III), the distances $\alpha_i$'s are decreasing over $i$. Putting these observations together we get 
\begin{equation*}
\|\bar{s}(K_i) - \bar{s}(K_j)\| 
	\,\le\, \|\bar{s}(K_i)\| + \|\bar{s}(K_j)\| 
	\,\le\,  \eta \alpha_i + \eta \alpha_j 
	\,\le\, 2 \eta \alpha_i
	\,\le\, 2 \eta D_p(K_i, K_{i+1}).
\end{equation*}
Thus, $\bar{s}$ is 2$\eta$-Lipschitz with respect to $D_p$.

\paragraph{Non-competitive.}
From the definition of $\bar{s}$ we have
\begin{equation*}
\|\bar{s}(K_i) - \bar{s}(K_{i+1})\| \,=\, \eta \alpha_i + \eta \alpha_{i+1} \,\ge\, \eta \alpha_i. 
\end{equation*}
This gives us the total movement
\begin{equation*}
\sum_{i=1}^\infty \|\bar{s}(K_i) - \bar{s}(K_{i+1})\|
	\ge \sum_{i=1}^\infty \eta \alpha_i
	=\infty,
\end{equation*}
where last step follows from \Cref{prop:hard-instance}(II). This concludes the proof. 
 \end{proof}

\subsection{Extension to a Full Selector}
\label{sec:full-selector}

We now extend the selector $\bar{s}$ from only the sets $\{K_i\}$ to all convex bodies to obtain the following more precise version of the statement of \Cref{thm:main-intro}.

\begin{theorem}[Restatement of \Cref{thm:main-intro}]  \label{thm:main}
For any $p \in (1,\infty]$, $d>\frac{4p-1}{p-1}$ and $\eps > 0$, there is a selector $s : \mathcal{K}^d \rightarrow \R^d$ that is $(5d + \eps)$-Lipschitz with respect to $D_p$ but is non-competitive. 
\end{theorem}	

When $p=\infty$, we can use known extension results (e.g. Theorem 18 of \cite{extensionSelector}) to extend our partial selector $\bar{s}$ to a selector whose Lipschitz constant is within a constant factor of optimality. However, the following fact shows that this is not possible for any $p<\infty$. 

\begin{restatable}{fact}{NoDpExt}
  \label{fct:no-dp-ext}
Fix $p \in [1,\infty)$ and $d > 2p+2$. Let $\hat{s}$ be the partial selector defined only on the unit ball $B$ with $\hat{s}(B) = e_1$ (which is trivially Lipschitz). Then $\hat{s}$ cannot be extended to a selector which is Lipschitz with respect to $D_p$.
\end{restatable}

We defer the proof to the appendix. In light of this fact, we extend \textbf{just our partial selector} $\bar s$ with a customized argument. As mentioned before, our argument relies on $\bar{s}(K_i)$ being deep within each $K_i$ (contrasted to the selector from Fact \ref{fct:no-dp-ext}, which is on the boundary) so that for a set $K \approx K_i$ it suffices for the extension to have $s(K) \approx s(K_i) = \bar{s}(K_i)$ to ensure that $s(K)$ is inside $K$. The sets $K$ distant from all $K_i$'s are actually easier to handle, since we can have $s(K)$ quite different from $s(K_i) = \bar{s}(K_i)$ without making $s$ non-Lipschitz.

At a high level, our final selector $s(K)$ is the following. Recall that the Steiner point selector $\st(\cdot)$ is $d$-Lipschitz with respect to $D_p$ for each $p\in [1,\infty]$ (see the remark following definition~(\ref{def:st-2})). Then:

\begin{itemize}\setlength{\itemsep}{0pt}
\item If $K$ is not close to any $K_i$, set $s(K) = \st(K)$. 				
\item If $K$ is close to some $K_i$, set $s(K)$ to be a linear interpolation between $\st(K)$ and $\bar{s}(K_i)$.
\end{itemize}
	
Formally, let $\cB_p(K, r) := \{K'\in \cK: D_p(K,K') \le r\}$ be the ball of radius $r$ centered at $K$ in the metric space $(\cK, D_p)$. As before, define $\alpha_i  := D_p(K_i, K_{i+1})$ and let $r_i := \frac{\alpha_i}{2}$. Define the following \emph{bump functions}:
\[f_i(K) := \max\left\{1 - \frac{D_p(K,K_i)}{r_i}, 0 \right\}.\]

Notice that $f_i$ takes value $1$ at $K_i$, value $0$ outside of $\cB_p(K_i,r_i)$, and interpolates linearly on $\cB_p(K_i,r_i)$.
Our extension is:
\begin{equation}
s(K) := \st(K) + \sum_i f_i(K) \cdot (\bar{s}(K_i)-\st(K)),
\end{equation}
where $\bar s$ is the partial selector of \Cref{prop:partial-selector} with $\eta = \frac{\eps}{4}$ and $\e > 0$ is sufficiently small (as a function of $d$ and $p$), and in particular $\e < 1$. Notice that nothing is lost in the statement of the theorem by only working with such $\e$.

To complete the proof of \Cref{thm:main}, we will show the following:
\begin{itemize}\setlength{\itemsep}{0pt}
	\item $s$ extends $\bar{s}$ (hence $s$ is non-competitive).
	\item $s$ is a selector, i.e. $s(K)\in K$ for all $K\in \cK$.
	\item $s$ is $(5d + \e)$-Lipschitz with respect to $D_p$.
\end{itemize}

\begin{lemma} \label{lem:disjoint-balls}
The balls $\cB_p(K_i, r_i)$ are pairwise disjoint.
\end{lemma}

\begin{proof}
	Again since the distances $\alpha_i$'s are decreasing over $i$ (\Cref{prop:hard-instance}(III)), for all $i<j$ we have 
$$r_i + r_j ~=~ \frac{\alpha_i}{2} + \frac{\alpha_j}{2} ~<~ \alpha_i ~=~ D_p(K_i, K_{i+1})~\le~ D_p(K_i, K_j).$$ Then by the triangle inequality the balls $\cB_p(K_i, r_i)$ and $\cB_p(K_j,r_j)$ are disjoint.  
 \end{proof}

\begin{cor}
The function $s$ is an extension of $\bar{s}$, i.e. $s(K_i) = \bar{s}(K_i)$ for all $i$.
\end{cor}
\begin{proof}
Since the support of $f_i$ is $\cB_p(K_i,r_i)$, the previous lemma gives that the functions $f_i$' have disjoint support. Since we also have $f_i(K_i) = 1$ (and hence $f_j(K_i) = 0$ for all $j \neq i$) we get from the definition of $s$ $$s(K_i) \,=\, \st(K) + (\bar{s}(K_i) - \st(K_i)) \,=\, \bar{s}(K_i),$$ as desired.
 \end{proof}	

Next, we show that $s$ is a selector. For that we will need the following comparison between the $D_p$ and the $D_{\infty}$ metrics.

	\begin{fact}{\cite[Corollary 2]{vitale1985lp}} \label{fact:compDp}
		Consider $d \ge 2$ and $p \in [0,\infty)$. Then for any sets $P,Q \in \mathcal{K}^d$ we have 
		\begin{align*}
			D_{\infty}(P,Q)^{1 + \frac{d-1}{p}} \,\le\, c(p,d) \cdot (\textrm{diam}(Q) + 2 D_\infty(P,Q))^{\frac{d-1}{p}} \cdot D_p(P,Q),
		\end{align*}
		where $\textrm{diam}(Q) = \max_{x,y \in Q} \|x-y\|$ is the diameter of $Q$ and $c(p,d)$ is a function of $p$ and $d$ only. 
	\end{fact}

\begin{lemma}
The function $s$ is a selector, i.e. $s(K)\in K$ for all $K$.
\end{lemma}
\begin{proof}
If $K$ is not in any $\cB_p(K_i,r_i)$, then $s(K) = \st(K) \in K$. Otherwise, suppose $K\in \cB_p(K_i,r_i)$ for some $i$. Again this implies that $f_j(K)=0$ for all $j\ne i$, hence 
\[s(K) = (1-f_i(K))\cdot \st(K) + f_i(K)\cdot \bar{s}(K_i).\] 
Furthermore, $f_i(K)\in [0,1]$, so $s(K)$ is a convex combination of $\st(K)$ and $\bar{s}(K_i)$. Since $K$ is convex, to obtain that $s(K) \in K$ it then suffices to show that $\bar{s}(K_i)\in K$. As noted in \Cref{prop:partial-selector}, $\alpha_i < 2\eps$. Using the bounds $\eta = \frac{\eps}{4}$ and $\eps<1$, this gives $\|\bar{s}(K_i)\| = \eta \alpha_i < \frac{1}{2}$. It then suffices to show that $K \supseteq B(0,\frac{1}{2})$. We have:
\begin{align*}
D_p(K, B) 
&\le D_p(K, K_i) + D_p(K_i, B) \tag{Triangle inequality}\\
&\le r_i + D_p(K_i, B)\\
&= \tfrac{1}{2}D_p(K_i, K_{i+1}) + D_p(K_i, B) \\
&\le \tfrac{3}{2}D_p(K_{i+1}, B) \tag{$K_{i+1}\subseteq K_i\subseteq B$}\\
&\le \tfrac{3}{2}D_\infty(K_{i+1}, B)\\
&\le \tfrac{3}{2}\eps \tag{\Cref{prop:hard-instance}(I)}.
\end{align*}
Together with \Cref{fact:compDp} we see that as long as $\e$ is small enough (as we have assumed) this implies that $D_{\infty}(K_{i+1}, B) \le \frac{1}{2}$.
But then Lemma~\ref{lem:contains-ball} implies that $K$ contains the ball $B(0,\frac{1}{2})$, which concludes the proof.
 \end{proof}

Having shown that $s$ is a selector and an extension of $\bar{s}$, we now turn to showing that it is Lipschitz with respect to $D_p$.

\begin{lemma}\label{lem:f-lipschitz}
For each $i$, the function $f_i$ is $\frac{1}{r_i}$-Lipschitz with respect to $D_p$.
\end{lemma}
\begin{proof}
Fix $P,Q$, and assume without loss of generality that $D_p(P,K_i) \ge D_p(Q,K_i)$. If $D_p(P,K_i) \ge r_i$, then
\begin{align*}
|f_i(Q) - f_i(P)| 
	= f_i(Q)
	&\le 1 - \frac{D_p(Q,K_i)}{r_i}\\
	&\le 1 - \frac{D_p(P,K_i) - D_p(P,Q)}{r_i} \tag{triangle inequality}\\
	&\le 1 - \frac{r_i - D_p(P,Q)}{r_i}\\
	&= \frac{D_p(P,Q)}{r_i}.
\end{align*}
Otherwise, 
	\begin{align*}
	|f_i(Q) - f_i(P)| 
	~&=~ \left(1 - \frac{D_p(Q,K_i)}{r_i}\right) - \left(1 - \frac{D_p(P,K_i)}{r_i}\right)\\
	~&=~  \frac{D_p(P,K_i) - D_p(Q,K_i)}{r_i}\\
	~&\le~ \frac{D_p(P,Q)}{r_i},
	\end{align*}
where the last step follows by the triangle inequality.
 \end{proof}
	
\begin{lemma}[$s$ is Lipschitz] \label{lem:lipschitz}
For all $P,Q$, 
\[ \|s(P) - s(Q)\| ~\le~ 
	(5d + \e) \cdot D_p(P,Q).\]
\end{lemma}
	
\begin{proof}
Recall that $\st(\cdot)$ is $d$-Lipschitz with respect to $D_p$ (see the note following definition~(\ref{def:st-2})). We first show that for each $i$, the function $g_i:= f_i(K)\cdot (\bar{s}(K_i) - \st(K))$ is ($2d+2\eta$)-Lipschitz with respect to $D_p$. Fix an $i$ and sets $P,Q\in \cK$ and suppose w.l.o.g.\ $D_p(Q,K_i) \le D_p(P,K_i)$. 
We have:
\begin{align*}
g_i(P) - g_i(Q)
	&= \big(f_i(P) - f_i(Q)\big)\cdot \bar{s}(K_i) + f_i(P) \cdot \st(P) - f_i(Q)\cdot \st(Q)\\
	&=\underbrace{\big(f_i(P) - f_i(Q)\big)\cdot \bar{s}(K_i)}_A 	+ \underbrace{f_i(P)\cdot \big(\st(P)-\st(Q) \big)}_B 
	\\
	&\hspace{0.5cm}+  \underbrace{\big(f_i(P)-f_i(Q)\big)\cdot \st(Q)}_C.
\end{align*}
We bound these terms separately. 

\paragraph{Term A.} Since $f_i$ is $\frac{1}{r_i}$-Lipschitz by Lemma~\ref{lem:f-lipschitz}, we have
\begin{align*}
\left \|\big(f_i(P) - f_i(Q)\big)\cdot \bar{s}(K_i)\right\|
	&\le \frac{D_p(P,Q)}{r_i} \cdot \|\bar{s}(K_i)\|\\
	&\le 2\eta \cdot D_p(P,Q) .
		\tag{since $\|\bar{s}(K_i)\| = \eta \alpha_i = 2\eta r_i$}
\end{align*}

\paragraph{Term B.}
Since $f_i(P)\le 1$, this term is at most $d\cdot D_p(P,Q)$ by Lipschitzness of $\st(\cdot)$.

\paragraph{Term C.}
If $D_p(Q,K_i) \ge r_i$, then $f_i(Q) = f_i(P) = 0$ and we are done. Otherwise, since $K_i$ is centrally symmetric, $\st(K_i) = 0$. Therefore we have 
\begin{align*}
\|\st(Q)\| 
	&= \|\st(Q) - \st(K_i)\|\\
	&\le d\cdot D_p(Q, K_i) \tag{Lipschitzness of $\st(\cdot)$}\\ 
	&\le d \cdot r_i.
\end{align*}
Using the $\frac{1}{r_i}$-Lipschitzness of $f_i$:
\[ 
\left \|\big(f_i(P) - f_i(Q)\big)\cdot \st(Q)\right\|
	~\le~ \frac{D_p(P,Q)}{r_i} \cdot \|\st(Q)\|
	~\le~ d\cdot D_p(P,Q).
\]
These three bounds together show that each $g_i$ is $(2d+2\eta)$-Lipschitz. 

We now show that $g:= \sum_i g_i$ is $(4d+4\eta)$-Lipschitz. Fix sets $P,Q$. Recall that the support of $g_i$ is $\cB_p(K_i,r_i)$, and these sets are disjoint. In particular, there are $i$ and $j$ such that $g(P) = g_i(P)$ and $g(Q) = g_j(Q)$. 
If $i=j$, then $\|g(P) - g(Q)\| \le (2d+2\eta) D_p(P,Q)$ since $g_i$ is $(2d+2\eta)$-Lipschitz. 
Otherwise, we have $g(P) = g_i(P) + g_j(P)$ and $g(Q) = g_i(Q)+g_j(Q)$. The function $g_i+g_j$ is $(4d+4\eta)$-Lipschitz, hence $\|g(P) - g(Q)\| \le (4d+4\eta) D_p(P,Q)$. 
Finally, since $\st(\cdot)$ is $d$-Lipschitz, it follows that $s = \st + g$ is $(5d+4\eta)$-Lipschitz. The proof follows by the choice of $\eta = \frac{\e}{4}$.
 \end{proof}


\subsection*{Acknowledgments}
We thank Boris Bukh for suggesting the question. 

{\small
\bibliographystyle{alpha}
\bibliography{lipschitz-chasing}
}


\newpage
\appendix

\section{Omitted proofs}

\subsection{Steiner Selector is Lipschitz for $D_1$} \label{app:SteinerD1}

\begin{fact}
	The Steiner selector $\st : \cK^d \rightarrow \R^d$ is $d$-Lipschitz with respect to the metric $D_1$.
\end{fact}

\begin{proof}
	For any two sets $A,B \in \cK^d$, by the definition of Steiner point we have
	\begin{align*}
    \|\st(A)-\st(B)\| &= \left\| d\cdot \int_{y\in S^{d-1}} y \cdot h_{A}(y) \ d\sigma(y)~-~ d\cdot \int_{y\in S^{d-1}} y \cdot h_{B}(y) \ d\sigma(y)\right\|\\
	&\le d\cdot \int_{y\in S^{d-1}} \Big\|(h_{A}(y)- h_{B}(y))\cdot y\Big\|\ d\sigma(y) \tag{Jensen's}\\
	&= d \int_{y\in S^{d-1}} \Big|h_A(y)- h_B(y)\Big|\ d\sigma(y) \tag{$\|y\|=1$}\\
	&= d\cdot  D_1(A,B),
\end{align*}
which concludes the proof. 
\end{proof}

\subsection{Being Lipschitz for $D_1$ Implies Competitiveness} \label{app:D1}

\begin{fact}
If $s:\cK^d \to \R^d$ is a $L$-Lipschitz selector with respect to $D_1$, then $s$ is an $L$-competitive selector.
\end{fact}

\begin{proof}
Let $B \supseteq K_1\supseteq K_2 \supseteq \dots$ be any nested sequence. Then since $s$ is $L$-Lipschitz with respect to $D_1$,
\begin{align*}
\sum_{t=1}^\infty \|s(K_t)-s(K_{t+1})\|	&\,\le\, \sum_{t=1}^\infty L\cdot  \int_{y\in S^{d-1}} \Big| h_{K_t}(y)- h_{K_{t+1}}(y)\Big|\ d\sigma(y) \\
	&= L\cdot \sum_{t=1}^\infty  \int_{y\in S^{d-1}} \Big(h_{K_t}(y)- h_{K_{t+1}}(y)\Big)\ d\sigma(y) \tag{$K_t\supseteq K_{t+1}$ so $h_{K_t} \ge h_{K_{t+1}}$}\\
	&= L \cdot \int_{y\in S^{d-1}} h_{K_1}(y) \ d\sigma(y) ~-~ d\cdot \lim_{t\to \infty} \int_{y \in S^{d-1}}  h_{K_t}(y)\ d\sigma(y)\\
	&\le L,
\end{align*}
where the last inequality follows from $h_{K_1}(y) \le 1$ and $h_{K_t}(y) + h_{K_t}(-y) \ge 0$. Hence the proof.
 \end{proof}

\subsection{No Lipschitz Selector Extension for $D_p$}

\NoDpExt*

\begin{proof}
Let $s$ be an arbitrary extension of $\hat{s}$. For $\theta \in [0, \frac{\pi}{2}]$, let $K_\theta := B - \bar{C}(e_1, \theta)$. We claim that 
\begin{equation}\label{eq:no-dp-ext}
\lim_{\theta \to 0^+} \frac{\|s(K_\theta) - s(B)\|}{D_p(K_\theta, B)} = \infty.
\end{equation}

First, we have $\|s(K_\theta) - s(B)\|\ge h_{B}(e_1) - h_{K_\theta}(e_1) = 1 - \cos(\theta) \ge \Omega(\theta^2)$ (the asymptotic $\Omega(\cdot)$ is as $\theta \rightarrow 0^+$). Also, for $y \in S^{d-1} - C(e_1, \theta)$ we have $h_B(y) = h_{K_\theta}(y)$, and for $y \in C(e_1,\theta)$ we have the bounds $0\le h_{K_\theta}(y) \le h_B(y) = 1$. In particular, $|h_{K_\theta}(y) - h_B(y)|^p \le 1$. Therefore,
\[ D_p(K_{y}, B) 
	= \left(\int_{y\in S^{d-1}}  |h_{K_\theta}(y) - h_B(y)|^p\ d\sigma(y)\right)^{\frac{1}{p}} 
	\le \big[\sigma(C(e_1,\theta))\big]^{\frac{1}{p}}
	\le O(\theta^{\frac{d-2}{p}}),
\]
where the last bound follows from Lemma~\ref{lem:area-cap} (again the asymptotic $O(\cdot)$ is with $\theta \rightarrow 0^+$). Choosing $d > 2p+2$ proves (\ref{eq:no-dp-ext}), and hence the fact.
 \end{proof}

\end{document}